\documentclass[11pt]{amsart}

\usepackage{latexsym,amssymb,amsfonts,amsmath,mathrsfs,hyperref,mathtools}

\calclayout

\usepackage{pifont}%
\usepackage{enumerate}
\hypersetup{colorlinks,citecolor=NavyBlue,filecolor=black,linkcolor=black,urlcolor=black}
\usepackage[dvipsnames]{xcolor}
\usepackage{comment} 
\usepackage[capitalize]{cleveref}

\newcommand{\vertiii}[1]{{\left\vert\kern-0.25ex\left\vert\kern-0.25ex\left\vert #1 
		\right\vert\kern-0.25ex\right\vert\kern-0.25ex\right\vert}}

  \DeclareMathOperator{\var}{Var}

  \DeclareMathOperator{\cb}{cb}
  \DeclareMathOperator{\fcb}{fcb}

  \newcommand{\norm}[1]{\|#1\|}

  \newcommand{\Id}{\ensuremath{\mathop{\rm Id}\nolimits}}

  \newcommand{\eps}{\varepsilon}

  \newcommand{\poly}{\mbox{\rm poly}}

  \newcommand{\ind}[1]{\mathbf{#1}}
   
  \newcommand{\beq}{\begin{equation}}
  \newcommand{\eeq}{\end{equation}}
  \newcommand{\beqn}{\begin{equation*}}
  \newcommand{\eeqn}{\end{equation*}}
  \newcommand{\beqr}{\begin{eqnarray}}
  \newcommand{\eeqr}{\end{eqnarray}}
  \newcommand{\beqrn}{\begin{eqnarray*}}
  \newcommand{\eeqrn}{\end{eqnarray*}}
  \newcommand{\bmline}{\begin{multline}}
  \newcommand{\emline}{\end{multline}}
  \newcommand{\bmlinen}{\begin{multline*}}
  \newcommand{\emlinen}{\end{multline*}}

  \theoremstyle{plain}
  \newtheorem{theorem}{Theorem}[section]
  
  \newtheorem{proposition}[theorem]{Proposition}

  \newtheorem{question}[theorem]{Question}
  \theoremstyle{definition}
  \newtheorem{definition}[theorem]{Definition}
  
  \newtheorem{conjecture}[theorem]{Conjecture}
  
  \theoremstyle{remark}
  \newtheorem{remark}[theorem]{Remark}
  
  \renewenvironment{proof}[1][]{
    	\begin{trivlist}
     	\item[\hspace{\labelsep}{\em\noindent Proof#1:\/}]}
     	{{\hfill$\Box$}
    	\end{trivlist}
  }
  
\newif\ifnotes\notesfalse

\ifnotes
\usepackage{color}
\definecolor{mygrey}{gray}{0.50}
\newcommand{\notename}[2]{{\textcolor{OliveGreen}{\footnotesize{\bf (#1:} {#2}{\bf ) }}}}
\newcommand{\noteswarning}{{\begin{center} {\Large WARNING: NOTES ON}\end{center}}}

\else

\newcommand{\notename}[2]{{}}
\newcommand{\noteswarning}{{}}

\fi

\usepackage{thm-restate} 

\begin{document}

\title[Influences of Fourier Completely Bounded Polynomials
]{
Influences of Fourier Completely Bounded Polynomials and Classical Simulation of Quantum Algorithms
}

\author[F. Escudero Guti\'errez]{Francisco Escudero Guti\'errez}
\address{CWI \& QuSoft, Science Park 123, 1098 XG Amsterdam, The Netherlands}
\email{feg@cwi.nl}

\thanks{\includegraphics[height=4ex]{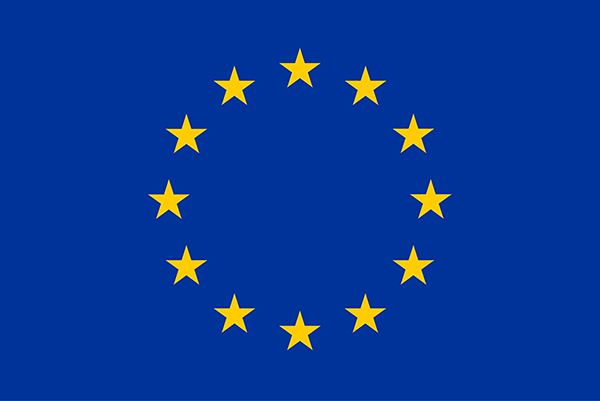} This research was supported by the European Union’s Horizon 2020 research and innovation programme under the Marie Sk{\l}odowska-Curie grant agreement no. 945045, and by the NWO Gravitation project NETWORKS under grant no. 024.002.003.}

\begin{abstract}

	We give a new presentation of the main result of Arunachalam, Bri\"et and Palazuelos (SICOMP'19) and show that quantum query algorithms are characterized by a new class of polynomials which we call Fourier completely bounded polynomials.
	We conjecture that all such polynomials have an influential variable.
	This conjecture is weaker than the famous Aaronson--Ambainis (AA) conjecture (Theory of Computing'14), 	but has the same implications for classical simulation of quantum query algorithms. 
	
	We prove a new case of the AA conjecture by showing that it holds for homogeneous Fourier completely bounded polynomials.
	This implies that if the output of $d$-query quantum algorithm is a homogeneous polynomial~$p$ of degree $2d$, then it has a variable with influence at least $\var[p]^2$.
	
	In addition, we give an alternative proof of the results of Bansal, Sinha and de Wolf (CCC'22 and QIP'23) showing that block-multilinear completely bounded polynomials have influential variables. Our proof is simpler, obtains better constants and does not use randomness. 
\end{abstract}

\maketitle

\noteswarning

\section{Introduction}\label{sec:intro}
Understanding the quantum query complexity of Boolean functions $f:D\to\{-1,1\}$, where $D$ is a subset of $\{-1,1\}^n$, has been a crucial task of quantum information science \cite{ambainis2018understanding}. Query complexity is a model of computing where rigorous upper and lower bounds can be proven, which allows one to compare the power of quantum computers with classical ones. Many celebrated quantum algorithms show an advantage in terms of query complexity, for example in unstructured search \cite{grover1996fast}, period finding \cite{shor1999polynomial}, Simon’s problem \cite{Simon}, NAND-tree evaluation \cite{farhi2007quantum} and element distinctness \cite{ambainis2007quantum}. However, these advantages are limited to be polynomial in the case of total functions (those with $D=\{-1,1\}^n$), while they can be exponential for highly structured problems (informally, this means that $|D| = o(2^n)$), such as for Simon's problem \cite{Simon}, period finding \cite{shor1999polynomial} or $k$-fold forrelation \cite{aaronson2015forrelation,tal2020towards,bansal2021k,sherstov2021optimal}. It is widely believed that a lot of structure is necessary for superpolynomial speedups\footnote{Recently, Yamakawa and Zhandry showed that superpolynomial speedups can be attained in unstructured search problems. However, that does not contradict the believe that structure is needed to achieve superpolynomial speedups in decision problems, which are those modeled by Boolean functions \cite{yamakawa2022verifiable}.}. The following folklore conjecture, which has circulated since the late 90s, but was first formally posed by Aaronson and Ambainis~\cite{aaronson2009need}, formalizes this idea. 

\begin{conjecture}[Folklore]\label{con:needforstructure}
	The output of $d$-query quantum algorithms can be simulated with error at most $\eps$ on at least a ($1-\delta$)-fraction of the inputs using poly($d,1/\eps,1/\delta$) classical queries.
\end{conjecture}

In other words, it is believed that quantum query algorithms can be approximated almost everywhere by classical query algorithms with only a polynomial overhead. 

A route towards proving Conjecture \ref{con:needforstructure} was designed by Aaronson and Ambainis using that the output of quantum query algorithms are polynomials\footnote{In this work we identify the output of quantum query algorithms on an input $x\in \{-1,1\}^n$ with the expectation of the quantum query algorithm on $x$, namely the difference of probability of accepting and the probability of rejecting.}. 
Indeed, Beals et al.\ \cite{polynomialmethod}, proved that the output of a $d$-query quantum algorithm is a bounded polynomial $p:\{-1,1\}^n\to\mathbb{R}$ of degree at most $2d$. Here, \emph{bounded} means that its supremum norm
\begin{equation}\label{eq:infnorm}
	\norm{p}_\infty=\sup_{x\in\{-1,1\}^n}|p(x)|
\end{equation}
is at most 1. Based on this observation, Aaronson and Ambainis conjectured in \cite{aaronson2009need} that every bounded polynomial of bounded degree has an influential variable.

\begin{conjecture}[Aaronson-Ambainis (AA)]\label{con:AAconjecture}
	Let $p:\{-1,1\}^n\to \mathbb{R}$ be a polynomial of degree at most $d$ with $\norm{p}_\infty\leq 1$. Then, $p$ has a variable with influence at least poly($\var[p],1/d$).
\end{conjecture}

The argument of \cite[Theorem 7]{aaronson2009need} to show that \cref{con:AAconjecture} would imply \cref{con:needforstructure} works as follows. Let $p$ the bounded polynomial of degree at most $2d$ that represents the output of $d$ query quantum algorithm. Say that we want to approximate $p(y)$ for some $y\in \{-1,1\}^n$. First, query an influential variable $i$ of $y$. Then, the restricted polynomial $p|_{x(i)=y(i)}$ would also be a bounded polynomial of degree at most $2d$, so we can query again an influential variable. Given that the influences of these variables are big, after a \emph{small} number of queries the remaining polynomial would have a low variance, so if we output its expectation it would be close to $p(y)$ with high probability.

Despite capturing the attention of a wide range of areas, little is known about Conjecture \ref{con:AAconjecture}. In fact, it was motivated by a similar result from analysis of Boolean functions, which proves the conjecture with inverse exponential dependence on the degree \cite{dinur2006fourier}. This result was reproved from a functional analytic perspective, highlighting its relation with the Bohnenblust-Hille inequality \cite{defant2019fourier}. 
A couple reductions to other conjectures have been made, such as that it is sufficient to prove the conjecture for one-block decoupled polynomials \cite{o2015polynomial}. Very recently, Lovett and Zhang  stated two conjectures related to fractional certificate complexity that, if true, would imply the AA conjecture \cite{lovett2022}. Also recently, Austrin et al.\ showed a connection of the AA conjecture with cryptography: if the AA conjecture is false,  then there is a secure key agreement in the quantum random oracle model that cannot be broken classically \cite{austrin2022impossibility}. However, we only know it to be true in a few particular cases: Boolean functions $f:\{-1,1\}^n\to\{-1,1\}$ \cite{midrijanis2005randomized,o2005every,Jain11}, symmetric polynomials \cite{Ivanishvili2019},  multilinear forms whose Fourier coefficients are all equal in absolute value \cite{montanaro2012some} and block-multilinear completely bounded polynomials \cite{bansal2022influence}. 

The last result is relevant in this context because  Arunachalam, Bri\"et and Palazuelos showed that quantum query algorithms output polynomials that are not only bounded, but also completely bounded \cite{QQA=CBF}. This is a more restricted normalization condition, which can be informally understood as the polynomial taking bounded values when evaluated not only on bounded scalars, but also on bounded matrix inputs. This way, one could try to use this extra condition to prove results about quantum query algorithms. 
The main motivation for Conjecture~\ref{con:AAconjecture} was the fact that it was known to hold with exponential dependence on the degree~\cite{dinur2006fourier}.

This idea was first put in practice by Bansal, Sinha and de Wolf \cite{bansal2022influence}. They showed that the AA conjecture holds for completely bounded block-multilinear forms, which implies an almost everywhere classical simulation result, similar to \cref{con:needforstructure}, for the amplitudes of certain quantum query algorithms. These algorithms query different (non-controlled) bit strings on every query, while Conjecture \ref{con:needforstructure} concerns algorithms that query the same controlled bit string on every query.

\subsection{Our results}
We follow that line of work and use the characterization of \cite{QQA=CBF} to design a route towards \cref{con:needforstructure}. Our first result is a new presentation of that characterization that is more convenient for our purposes. To do this we introduce the Fourier completely bounded $d$-norms ($\norm{\cdot}_{\fcb,d}$), which are relaxations of the supremum norm. In these norms we not only take the supremum of the values that the polynomial takes over Boolean strings as in \cref{eq:infnorm}, but also on matrix inputs that behave like Boolean strings. We will not include formal definitions in the introduction, but we illustrate the concept of having \emph{Boolean behavior of degree $d$} with an example. For $m\in \mathbb{N}$, we denote the $m\times m$ real matrices by $M_m$. Say that $d=4$ and $n=6$, then if a pair of vectors $u,v\in\mathbb{R}^m$  and a string of matrices $A\in (M_m)^6$ have Boolean behaviour of degree $4$, they satisfy, for instance, $$\langle u,A(1)A(1)A(2)A(3)v\rangle=\langle u, A(5)A(2)A(3)A(5)v\rangle,$$
because they should simulate the relation $x(1)x(1)x(2)x(3)=x(5)x(2)x(3)x(5)$ satisfied by any Boolean string $\{-1,1\}^6$. As the reader might guess, $(u,v,A)$ will have Boolean behavior of degree $d$ if it simulates the relations of $\mathbb{F}_2^n$ that involve product of $d$ of the canonical generators. 

Using the Fourier expansion of polynomials defined on the Boolean hypercube we will introduce a natural way of evaluating polynomial on matrix inputs that have Boolean behavior, which allows us to introduce  the Fourier completely bounded $d$-norm.

\begin{definition}(Informal version of \cref{def:fcbnorm})
	Let $p:\{-1,1\}^n\to \mathbb{R}$ be a polynomial of degree at most $d$. Its \emph{Fourier completely bounded $d$-norm} is given by
	\begin{align}\label{eq:DefFCB}
		\norm{p}_{\fcb,d}:=\sup|p(u,v,A)|
	\end{align}
	where the supremum is taken over all $(u,v,A)$ that have Boolean behavior of degree $d$.
\end{definition} 

After a reinterpretation of the semidefinite programs proposed in \cite{gribling2019semidefinite} to characterize quantum query complexity, based on \cite{QQA=CBF}, we show that the Fourier completely bounded $d$-norms are those that characterize quantum query algorithms. 

\begin{restatable}{theorem}{QQAareFCBPolynomials}\label{theo:QQAareFCBPolynomials}
	Let $p:\{-1,1\}^n\to \mathbb{R}$. Then, $p$ is the output of a $d$-query quantum algorithm if and only if its degree is at most $2d$ and $\norm{p}_{\fcb,2d}\leq 1$. 
\end{restatable}

This new presentation of the main result of \cite{QQA=CBF} is more compact than the original one. It is presented directly in terms of polynomials of the Boolean hypercube, does not involve a minimization over possible completely bounded extensions of $p$, and eludes the use of tensors. 

Given that the Fourier completely bounded $d$-norms are at least the supremum norm\footnote{From the results of \cite{briet2018failure} it can be inferred that there is a sequence of polynomials $p_n$ of degree 3 such that $\norm{p_n}_{\fcb,3}/\norm{p_n}_{\infty}\to_{n} \infty.$}, Theorem \ref{theo:QQAareFCBPolynomials} suggests that Conjecture \ref{con:AAconjecture} may be more general than necessary. Hence, we propose the following weaker conjecture, that would also imply Conjecture \ref{con:needforstructure}.

\begin{conjecture}\label{con:AAconjectureFCB}
	Let $p:\{-1,1\}^n\to \mathbb{R}$ be a polynomial of degree at most $d$ with $\norm{p}_{\fcb,d}\leq 1$. Then, $p$ has a variable with influence at least poly($\var[p],1/d$).
\end{conjecture}

Using a generalization through \emph{creation} and \emph{annihilation} operators of the construction used by Varopoulos to rule out a von Neumann's inequality for degree 3 polynomials \cite{Varopoulos:1974}, we can prove a particular case of Conjecture \ref{con:AAconjectureFCB}. 

\begin{restatable}{theorem}{AAFCBconjectureForMaxDegreePol}\label{theo:AAFCBconjectureForMaxDegreePol}
 	Let $d\in \mathbb{N}$. Let $p:\{-1,1\}^n\to \mathbb{R}$ be a homogeneous polynomial of degree $d$ and with $\norm{p}_{\fcb,d}\leq 1$. Then, the maximum influence of $p$ is at least $\var [p]^2$.
\end{restatable}

The proof of the homogeneous case does not straightforwardly generalize (see \cref{rem:whynot}), but it suggests a way to solve the general case (see \cref{rem:howtogeneralize}). In particular, we propose Question \ref{que:generalization} (that reminds of tensor networks and almost-quantum correlations), which if answered affirmatively would imply \cref{con:AAconjectureFCB}.

\cref{theo:AAFCBconjectureForMaxDegreePol} is the first result concerning the AA conjecture whose constant has no dependence on the degree (to prove \cref{con:needforstructure} we could afford a polynomial dependence on the degree). Also, it requires considerably fewer algebraic constraints than the other particular cases for which we know AA conjecture to hold. In addition,  thanks to \cref{theo:QQAareFCBPolynomials}, it can be interpreted directly in terms of quantum query algorithms.

\begin{restatable}{corollary}{AAFCBconjectureForMaxDegreePolv2}\label{theo:AAFCBconjectureForMaxDegreePolv2}
	Let $d\in \mathbb{N}$. Let $\mathcal{A}$ be a $d$-query quantum algorithm whose output is a homogeneous polynomial $p:\{-1,1\}^n\to \mathbb{R}$ of degree $2d$. Then, the maximum influence of $p$ is at least $\var [p]^2$.
\end{restatable}

With a similar construction as the one we used for \cref{theo:AAFCBconjectureForMaxDegreePol}, we can reprove the results of \cite{bansal2022influence} regarding the influence of block-multilinear completely bounded polynomials. These polynomials have a particular algebraic structure and also a normalization condition when evaluated on matrix inputs (see \cref{subsec:AAforBCB} below). 

\begin{restatable}{theorem}{AAforBCB}\label{theo:AAforBCB}
	Let $d\in\mathbb{N}$. Let $p:\{-1,1\}^{n\times d}\to \mathbb{R}$ be a block-multilinear degree $d$ polynomial with $\norm{p}_{\cb}\leq 1$. Then, $p$ has a variable of influence at least $(\var [p]/d)^2$. What is more, if $p$ is homogeneous of degree $d$, then it has a variable of influence at least $\var [p]^2$.   
\end{restatable}
\cref{theo:AAforBCB} corresponds to \cite[Theorem 1.4]{bansal2022influence}, where  Bansal et al.\ proved the same result but with influences at least $\var [p]^2/[e(d+1)^4]$ in the general case and with $\var [p]^2/(d+1)^2$ in the homogeneous degree $d$ case. Their proofs involve evaluating $p$ in \emph{random infinite dimensional matrix inputs}, which they can control using ideas of free probability. However, our proof evaluates $p$ in explicit finite dimensional matrix inputs, is shorter and obtains better constants. In particular, our constant for the homogeneous case is optimal.

\subsection{Some notation} Given $n\in \mathbb{N}$, $[n]$ denotes the set $\{1,\dots,n\}$, $[n]_0$ the set $\{0,1,\dots,n\}$, $M_n$ the space of $n\times n$ real matrices, $B_n\subset M_n$ the space of $n\times n$  real matrices with operator norm (largest singular value) at most $1$, and $S^{n-1}\subset \mathbb{R}^n$ the set of vectors with euclidean norm 1. Given, $d,n\in \mathbb N$, we use $\ind i$ to denote a multi-index in 
$[n]^d$. 

\section{The Fourier completely bounded $d$-norms}
There is a vast theory concerning the properties of multilinear maps $T:\mathbb{R}^n\times\dots\times\mathbb{R}^n\to \mathbb{R}$ that are completely bounded, i.e., bounded when they are extended to matrix domains \cite{paulsenoperatoralgebras}. However, to the best of our knowledge, there is no notion of being \emph{completely bounded} for polynomials $p:\{-1,1\}^n\to\mathbb{R}$ defined on the Boolean hypercube. Here, we propose a matrix notion of behaving like a Boolean string. Then, using the Fourier expansion of these polynomials we define the evaluation of the polynomials on these matrix inputs that behave like Boolean strings. Finally, we introduce the Fourier completely bounded $d$-norms and prove a few of their properties. 

Every $p:\{-1,1\}^n\to\mathbb{R}$ can be written as
\begin{equation}\label{eq:FourierExpansion}
	p(x)=\sum_{S\subset [n]}\hat{p}(S)\prod_{i\in S}x(i),
\end{equation}
where $\hat{p}(S)$ are the Fourier coefficients of $p$. We say that $p$ has degree at most $d$ if $\hat{p}(S)=0$ for every $|S|>d$, where $|S|$ denotes the cardinality of $S$. 

We will be interested on simulating the behavior of bit strings $x\in\{-1,1\}^n\times \{1\}$ with one extra frozen variable\footnote{The extra variable set to $1$ is there because quantum query algorithms query a controlled bit string, and not a non-controlled version, which would not require that extra variable.}. Given $d\in\mathbb{N}$ and $\ind{i},\ind{j}\in [n+1]^d$ we say that $\ind{i}\sim\ind{j}$, if 
\begin{equation}\label{eq:BBscalars}
	x(i_1)\dots x(i_d)=x(j_1)\dots x(j_d)\ \mathrm{for\ every}\ x\in \{-1,1\}^n\times\{1\}.
\end{equation}
In other words, if we define $$S_{\ind{i}}:=\{k\in [n]: \ k\text{ occurs an odd number of times in }\ind{i}\},$$ then $\ind i \sim \ind j$ if and only if $S_{\ind i}=S_{\ind j}$. Note that $n+1$ does not belong to these sets $S_{\ind i}$. Given $S\subset [n]$ with $|S|\leq d$, we write $[\ind i^S]$ to denote the equivalence class of indices $\ind i$ such that $S_{\ind{i}}=S$.

\begin{definition}
	Let $n,d,m\in\mathbb{N}$. Let $u,v\in S^{m-1}$ and let $A\in (B_m)^n$. We say that $(u,v,A)$ has \emph{Boolean behavior of degree d} if $$\langle u,A(i_1)\dots A(i_d)v\rangle=\langle u,A(j_1)\dots A(j_d)v\rangle$$ for all $\ind{i},\ind{j}\in [n+1]^d$ such that $\ind i\sim\ind j$. We call $\mathscr{BB}^d$ to the set of $(u,v,A)$ with Boolean behavior of degree $d$. 
\end{definition}

Informally, having Boolean behavior of degree $d$ means that the relations of \cref{eq:BBscalars} and some normalization conditions are satisfied. In particular, for any bit string $x\in \{-1,1\}^n\times \{1\}$ and any $d\in \mathbb{N}$, we have that $(1,1,x)$ has Boolean behavior of degree $d$.

Also note that given $d\in \mathbb{N}$, for every $S\in [n]$ with $|S|\leq d$ there is at least one $\ind i \in [n+1]^d$ such that $S_{\ind{i}}=S.$  Thus, given $(u,v,A)$ with Boolean behavior of degree $d$, for every $|S|\leq d$ the product $\prod_{i \in S} x(i)$ can be simulated (in a unique manner) by $\langle u,A(i_1^S)\dots A(i_d^S)v\rangle$. In particular, this means that for a polynomial $p$ of degree at most $d$, we can define through \cref{eq:FourierExpansion} an evaluation of $p$ on every $(u,v,A)$ that has Boolean behavior of degree $d$, which leads to the definition Fourier completely bouded $d$-norm.

\begin{definition}\label{def:fcbnorm}
	Let $p:\{-1,1\}^n\to\mathbb{R}$ be a polynomial of degree at most $d$. Then, its \emph{Fourier completely bounded $d$-norm} is defined by
	$$\norm{p}_{\fcb,d}=\sup_{(u,v,A)\in \mathscr{BB}^d}\sum_{S\subset [n],|S|\leq d}\hat{p}(S)\langle u,A(i_1^S)\dots A(i_d^S)v\rangle.$$
\end{definition}

The rest of the section is devoted to prove a few results concerning the Fourier completely bounded $d$-norms. First of all we show that, indeed, they are norms. 
\begin{proposition}
	Let $d\in\mathbb{N}$. Then, $\norm{\cdot}_{\fcb,d}$ is a norm in the space of polynomials $p:\{-1,1\}^n\to \mathbb{R}$ of degree at most $d$.
\end{proposition}
\begin{proof}
	It clearly satisfies the triangle inequality and is homogeneous. Also, if $p=0$ then $\norm{p}_{\fcb,d}=0$, and vice versa, because $\norm{p}_{\infty}\leq \norm{p}_{\fcb,d}$. 
\end{proof} 

One nice property of these norms is that they can be computed as semidefinite programs (SDPs), which are optimization problems involving linear positive semidefinite constraints whose value can be efficiently approximated (see \cite{laurent2005semidefinite} for an introduction to SDPs). 

\begin{proposition}\label{prop:fcbasanSDP}
	Let $p:\{-1,1\}^n\to\mathbb{R}$ be a polynomial of degree at most $d$. Then, its Fourier completely bounded $d$-norm can be written as the following SDP
	\begin{align}
		\label{eq:fcbsdp}\norm{p}_{\fcb,d}=& \sup \sum_{S\in [n],|S|\leq d}\hat{p}(S)\langle u,v_{\ind{i}^S}\rangle,\\
		\nonumber& u,v,v_{\ind{i}}\in \mathbb{R}^m,\ m\in \mathbb{N},\ \ind{i}\in [n+1]^s,\ s\in [d],\\
		\label{eq:BBinSDP}& \langle u,v_{\ind{i}}\rangle=\langle u,v_{\ind{j}}\rangle,\ \mathrm{if}\ \ind i\sim \ind j,\ \ind i,\ind j\in [n+1]^d,\\
		\label{eq:unitvectinSDP}&\langle u,u\rangle=\langle v,v\rangle=1,\\
		\label{eq:contrinSDP}& \mathrm{Gram}_{\substack{\ind{j}\in [n+1]^{s}, \\s\in [d-1]_0}}\{v_{i\ind{j}} \}\preccurlyeq \mathrm{Gram}_{\substack{\ind{j}\in [n+1]^{s}, \\s\in [d-1]_0}}\{v_{\ind{j}}\},\ \mathrm{for}\ i\in [n+1],
	\end{align}
	where we by $v_\ind{j}$ with $\ind j\in [n+1]^0$ we mean $v$, $\mathrm{Gram}$ denotes the gram matrix and the symbol `$\preccurlyeq$' the usual matrix inequality. 
\end{proposition}
\begin{proof}
	Let $\vertiii{p}$ be the expression on the right-hand side of \cref{eq:fcbsdp}. Note that \cref{eq:BBinSDP} represents the relations of bit strings of \cref{eq:BBscalars}, while \cref{eq:unitvectinSDP,eq:contrinSDP} encode normalization conditions. 
	
	On the one hand, every $(u,v,A)\in \mathscr{BB}^d$ defines a feasible instance for $\vertiii{p}$ through $$v_{\ind{i}}:=A(i_1)\dots A(i_s)v$$
	for every $\ind{i}\in [n+1]^s$ and every $s\in [d]$. Given that the value of this instance is $$\sum_{S\subset [n],|S|\leq d}\hat{p}(S)\langle u,A(i_1^S)\dots A(i_d^S)v\rangle$$ we have that $\vertiii{p}\geq \norm{p}_{\fcb,d}.$
	
	On the other hand, let $u,v,v_{\ind{i}}\in\mathbb{R}^{m}$ be a feasible instance of $\vertiii{p}$. For $i\in [n+1]$ define $A(i)\in M_m$ as the linear map from $\mathbb{R}^m$ to $\mathbb{R}^m$ that takes $v_{\ind{j}}$ to $v_{i\ind{j}}$ for every $\ind j\in [n+1]^{s}$ and every $s\in [d-1]_0$, and it is extended to $0$ to the rest of $\mathbb{R}^m$. First of all, we should check that this is a correct definition. In other words, we should check that for every $\lambda\in \mathbb{R}^M$, with $M=(n+1)^{d-1}+\dots+(n+1)^0$, we have that $$\sum_{\ind{j}}\lambda_{\ind j}v_{\ind{j}}=0\ \implies\  \sum_{\ind{j}}\lambda_{\ind j}v_{i\ind{j}}=0.$$
	Indeed, we can prove something stronger: 
	\begin{align*}
		(\sum_{\ind{j}}\lambda_{\ind j}v_{i\ind{j}})^T\sum_{\ind{j}'}\lambda_{\ind j'}v_{i\ind{j}'}&=\lambda^T\mathrm{Gram}_{\substack{\ind{j}\in [n+1]^{s}, \\s\in [d-1]_0}}\{v_{i\ind{j}} \}\lambda\leq\lambda^T \mathrm{Gram}_{\substack{\ind{j}\in [n+1]^{s}, \\s\in [d-1]_0}}\{v_{\ind{j}}\}\lambda\\ 
		&=(\sum_{\ind{j}}\lambda_{\ind j}v_{\ind{j}})^T\sum_{\ind{j}'}\lambda_{\ind j'}v_{\ind{j}'}.
	\end{align*}
	The above calculation also proves that the $A(i)$'s are contractions, and thanks to \cref{eq:BBinSDP} it follows that $(u,v,A)$ has Boolean behavior of degree $d$. Finally, note that the value of this $(u,v,A)$ for $\norm{p}_{\fcb,d}$ is the same as the value of $(u,v,v_{\ind{j}})$ for $\vertiii{p}$, so $ \norm{p}_{\fcb,d}\geq \vertiii{p}.$
\end{proof}

Given $d, d'\in\mathbb{N}$ with $d'>d$ and a polynomial $p:\{-1,1\}^n\to\mathbb{R}$ of degree at most $d$, $\norm{p}_{\fcb,d}$ and $\norm{p}_{\fcb,d'}$ have different definitions, but they are comparable. In particular, we prove that the Fourier completely bounded $d$-norms are not increasing.

\begin{proposition}\label{prop:fcbD+1<fcbD}
	Let $p:\{-1,1\}^n\to\mathbb{R}$ be a polynomial of degree at most $d$. Then, $$\norm{p}_{\fcb,d+1}\leq \norm{p}_{\fcb,d}.$$
\end{proposition}

\begin{remark}
	 \cref{prop:fcbD+1<fcbD} is coherent with \cref{theo:QQAareFCBPolynomials} (proved below), because allowing more queries to quantum algorithms only increases their power. \cref{theo:QQAareFCBPolynomials} also suggests that $\norm{p}_{\fcb,n}=\norm{p}_{\infty}$ should hold, because $n$ quantum queries should be enough to output any bounded polynomial. If true, alongside \cref{prop:fcbasanSDP,prop:fcbD+1<fcbD}, it would mean that $(\norm{p}_{\fcb,d})_{d\in [n]}$ is a decreasing hierarchy of SDPs that tend to $\norm{p}_{\infty}$.	
\end{remark}

\begin{proof}[ of \cref{prop:fcbD+1<fcbD}]
	Let $(u,v,A)$ have Boolean behavior of degree $d+1$. Then, $$(\tilde{u},\tilde{v},\tilde{A})=(u,\frac{A(n+1)v}{\|A(n+1)v\|},A)$$ has Boolean behavior of degree $d$. Also, given that $d+1>d$, we have that for every $S\subset [n]$ with $|S|\leq d$, there exists $\ind{i}\in [n+1]^{d+1}$ such that $S_\ind{i}=S$ and $i_{d+1}=n+1$ and  
	\begin{equation}\label{eq:fcbD+1<fcbD}
		\langle u, A(i_1)\dots A(i_{d+1})v\rangle=\norm{A(n+1)v}\langle \tilde{u}, \tilde{A}(i_1)\dots \tilde{A}(i_d)\tilde{v}\rangle 
	\end{equation}  
	$$ $$
	This way, 
	\begin{align*}
		\norm{p}_{\fcb,d+1}&=\sup_{(u,v,A)\in\mathscr{BB}^{d+1}} \sum_{S\subset [n],|S|\leq d}\hat{p}(S)\langle u, A(i^S_1)\dots A(i^S_{d+1})v\rangle\\
		&= \norm{A(n+1)v}\sup_{(u,v,A)\in\mathscr{BB}^{d+1}\ }\sum_{S\subset [n],|S|\leq d}\hat{p}(S)\langle \tilde{u}, \tilde{A}(i^S_1)\dots \tilde{A}(i^S_d)\tilde{v}\rangle \\
		&\leq \sup_{(\tilde{u},\tilde{v},\tilde{A}(i))\in\mathscr{BB}^d\ }\sum_{S\subset [n],|S|\leq d}\hat{p}(S)\langle \tilde{u}, \tilde{A}(i^S_1)\dots \tilde{A}(i^S_d)\tilde{v}\rangle\\
		&= \norm{p}_{\fcb,d},
	\end{align*}
	where in the equality we have used \cref{eq:fcbD+1<fcbD}, and in the inequality that $(\tilde{u},\tilde{v},\tilde{A})$ has Boolean behavior of degree $d$ and $\norm{A(n+1)v}\leq 1$.
\end{proof}

The next proposition states that $\norm{\cdot}_{\fcb,d}$ does not increase after restrictions, which is a relevant feature to ensure that Conjecture \ref{con:AAconjectureFCB} implies Conjecture \ref{con:needforstructure}. Given a polynomial $p:\{-1,1\}^n\to \mathbb{R}$ and $i\in [n]$, the \emph{restriction of $p$ to the $i$-th variable being set to $y\in\{-1,1\}$} is the polynomial $q:\{-1,1\}^{n-1}\to \mathbb{R}$ (whose variables we index with $x(1),\dots,x(i-1),x(i+1),\dots,x(n)$ for convenience) defined by $q(x):=p(x(1),\dots,x(i-1),y,x(i+1),\dots,x(n))$.
\begin{proposition}\label{prop:noincreaseafterrest}
	Let $p:\{-1,1\}^n\to \mathbb{R}$ be a polynomial of degree at most $d$ and let $i\in [n]$. Let $q:\{-1,1\}^{n-1}\to \mathbb{R}$ be the restriction of $p$ to the $i$-th variable being set to $y\in\{-1,1\}$. Then, 
	\begin{equation*}
		\norm{q}_{\fcb,d}\leq \norm{p}_{\fcb,d}.
	\end{equation*}
\end{proposition}
\begin{proof}
	Let $(u,v,A(1),\dots,A(i-1),A(i+1),\dots A(n+1))$ with Boolean behavior of degree $d$. Define $\tilde{u}:=u$, $\tilde{v}:=v$ and  $\tilde{A}(j)$ for $j\in [n+1]$ as $$\tilde{A}(j)=\left\{\begin{array}{ll}
		A(j) & \mathrm{if\ }j\neq i,\\
		yA(n+1) & \mathrm{if\ }j=i.
	\end{array}\right.$$
	Then, $(\tilde{u},\tilde{v},\tilde{A}(1),\dots,\tilde{A}(n+1))$ has Boolean behavior of degree $d$. Now note that for every $S\subset[n+1]-\{i\}$, we have that
	\begin{equation}\label{eq:rest1}
		\hat{q}(S)=\hat{p}(S)+y\hat{p}(S\cup\{i\}).
	\end{equation} 
	Also, for every $S\subset [n+1]-\{i\}$ with $|S|\leq d-1$, it is satisfied that 
	\begin{equation}\label{eq:rest2}
		\langle \tilde{u}, \tilde{A}(j^S_1)\dots \tilde{A}(j^S_d)\tilde{v}\rangle=y\langle \tilde{u}, \tilde{A}(j^{S\cup\{i\}}_1)\dots \tilde{A}(j^{S\cup\{i\}}_d)\tilde{v}\rangle.
	\end{equation} 
	Thus, 
	\begin{align*}
		\norm{q}_{\fcb,d}&=\sup_{\substack{(u,v,A(j))\in\mathscr{BB}^d\\ j\in[n+1]-\{i\}}} &&\sum_{S\subset [n+1]-\{i\},|S|\leq d}\hat{q}(S)\langle u, A(j^S_1)\dots A(j^S_{d})v\rangle\\
		&=\sup_{\substack{(u,v,A(j))\in\mathscr{BB}^d\\ j\in[n+1]-\{i\}}}&&\sum_{S\subset [n+1]-\{i\},|S|\leq d}\hat{p}(S)\langle u, A(j^S_1)\dots A(j^S_d)v\rangle \\
		&+&&\sum_{S\subset [n+1]-\{i\},|S|\leq d-1}y\hat{p}(S\cup\{i\})\langle u, A(j^S_1)\dots A(j^S_d)v\rangle \\
		&\leq\sup_{\substack{(\tilde{u},\tilde{v},\tilde{A}(j))\in\mathscr{BB}^d\\ j\in [n+1]}}&&\sum_{S\subset [n+1]-\{i\},|S|\leq d}\hat{p}(S)\langle \tilde{u}, \tilde{A}(j^S_1)\dots \tilde{A}(j^S_d)\tilde{v}\rangle \\
		&+&&\sum_{S\subset [n+1]-\{i\},|S|\leq d-1}\hat{p}(S\cup\{i\})\langle \tilde{u}, \tilde{A}(j^{S\cup\{i\}}_1)\dots \tilde{A}(j^{S\cup\{i\}}_d)\tilde{v}\rangle \\
		&=\sup_{\substack{(\tilde{u},\tilde{v},\tilde{A}(j))\in\mathscr{BB}^d\\ j\in [n+1]}}&&\sum_{S\subset [n],|S|\leq d}\hat{p}(S)\langle \tilde{u}, \tilde{A}(j^S_1)\dots \tilde{A}(j^S_d)\tilde{v}\rangle \\
		&= \norm{p}_{\fcb,d},
	\end{align*}
	where in the second equality we have used \cref{eq:rest1} and in the first inequality \cref{eq:rest2}.
\end{proof}

\section{Quantum query algorithms are Fourier completely bounded polynomials}
Now we are ready to prove  \cref{theo:QQAareFCBPolynomials}, that fully characterizes quantum query algorithms in terms of the Fourier completely bounded $d$-norms.  

\QQAareFCBPolynomials*

To prove \cref{theo:QQAareFCBPolynomials} we just have to reinterpret the semidefinite programs of \cite{gribling2019semidefinite}, based on \cite{QQA=CBF}.  

\begin{theorem}[Gribling-Laurent] \label{theo:sdpcbnorm}
	Let $p:\{-1,1\}^n\to \mathbb{R}$. Then, $p$ is the output of $d$-query quantum algorithm if and only if its degree is at most $2d$ and the value of the following semidefinite program is at most 0,
	\begin{align}
		\max&\ -w+\sum_{x\in \{-1,1\}^n} \frac{p(x)\phi(x)}{2^n}\label{eq:sdpcbnorm}\\
		\mathrm{s.t.}&\ w\geq 0,\ m\in\mathbb{N},\ A_s\in (B_m)^{n+1},\  u,v \in \mathbb{R}^m\nonumber,\ s\in [2d],\\
		&\ \norm{\phi}_1=1,\  \norm{u}^2=\norm{v}^2=w,\nonumber\\
		&\ \hat{\phi}(S_\ind{i})=\langle u,A_1(i_1)\dots A_{2d}(i_{2d})v\rangle,\ \ind{i}\in [n+1]^{2d} \nonumber,
	\end{align}  
	where $\norm{\phi}_1=\sum_{x\in\{-1,1\}^n}\frac{|\phi(x)|}{2^n}$.
\end{theorem}

\begin{remark}\label{rem:sdpcbnorm}
	\cref{theo:sdpcbnorm} corresponds to \cite[Equation (24)]{gribling2019semidefinite}. There, the authors not only ask for the $A_s(i)$ to be contractions, but also unitaries. However, that extra restriction does not change the value of the semidefinite program because we can always block-encode a contraction in the top left corner of an unitary (see for instance \cite[Lemma 7]{Aaronson2015PolynomialsQQ}). We also want to remark that $A_s(i)$ can be taken to be equal to $A_{s'}(i)$ for every $s,s'\in [2d]$ and every $i\in [n+1]$, as this extra restriction does not change value of the semidefinite program. Indeed, let $(u,v,A_s,w,\phi)$ be part of feasible instance of \cref{eq:sdpcbnorm}. Define now 
	\begin{align*}
		\tilde{u}&:=u\otimes e_1,\\
		\tilde{v}&:=v\otimes e_{2d+1},\\
		A(i)&:=\sum_{s\in [2d]}A_s(i)\otimes e_{s}e_s^T,
	\end{align*}
	where $\{e_s\}_{s\in [2d+1]}$ is an orthonormal basis of $\mathbb{R}^{2d+1}$. Then, $$\langle u,A_1(i_1)\dots A_d(i_{2d})v\rangle=\langle \tilde{u},\tilde{A}(i_1)\dots \tilde{A}(i_{2d})\tilde{v}\rangle,$$
	for every $\ind{i}\in [n+1]^{2d}$. Hence, $(\tilde{u},\tilde{v},\tilde{A},w,\phi)$ is a feasible instance for \cref{eq:sdpcbnorm} that attains the same value as $(u,v,A_s,w,\phi)$. 
\end{remark}

\begin{proof}[ of \ref{theo:QQAareFCBPolynomials}]
	Thanks to \cref{theo:sdpcbnorm} and \cref{rem:sdpcbnorm}, we know that $p$ is the output of $d$-query quantum algorithm if and only if its degree is at most $2d$ and the following constraint is satisfied
	\begin{align}
		&\ \sum_{x\in \{-1,1\}^n} \frac{p(x)\phi(x)}{2^n}\leq w\label{eq:sdpcbnorm2}\\
		\mathrm{s.t.}&\ w\geq 0,\ m\in\mathbb{N},\ A\in (B_m)^{n+1},\ u,v \in \mathbb{R}^m\nonumber,\\
		&\ \norm{\phi}_1=1,\label{eq:phi1norm=1}\\
		&\ \norm{u}^2=\norm{v}^2=w,\nonumber\\
		&\ \hat{\phi}(S_\ind{i})=\langle u,A(i_1)\dots A(i_{2d})v\rangle,\ \ind{i}\in [n+1]^{2d} \nonumber.
	\end{align} 
	Now, note that if $(u,v,A,\phi,w)$ satisfies all conditions of \cref{eq:sdpcbnorm2} except for \cref{eq:phi1norm=1}, then $(u/\sqrt{\norm{\phi}_1},$ $v/\sqrt{\norm{\phi}_1},A,\phi/\norm{\phi}_1,w/\norm{\phi}_1)$ would be a feasible instance. Furthermore, given that $$\sum_{x\in \{-1,1\}^n} \frac{p(x)\phi(x)}{2^n}\leq w \iff \frac{1}{\norm{\phi}_1}\sum_{x\in \{-1,1\}^n} \frac{p(x)\phi(x)}{2^n}\leq \frac{w}{\norm{\phi}_1},$$
	we can write \cref{eq:sdpcbnorm2} forgetting about the normalization condition of \cref{eq:phi1norm=1}. In other words, \cref{eq:sdpcbnorm2} is equivalent to
	\begin{align}
		&\ \sum_{x\in \{-1,1\}^n} \frac{p(x)\phi(x)}{2^n}\leq w\label{eq:sdpcbnorm3}\\
		\mathrm{s.t.}&\ w\geq 0,\ m\in\mathbb{N},\  A\in (B_m)^{n+1},\ u,v \in \mathbb{R}^m\nonumber,\\
		&\ \norm{u}^2=\norm{v}^2=w,\\
		&\ \hat{\phi}(S_\ind{i})=\langle u,A(i_1)\dots A(i_{2d})v\rangle,\ \ind{i}\in [n+1]^{2d} \nonumber.
	\end{align} 
	In addition, by homogeneity we can assume $w=1$. Also note, that if $(u,v,A)$ are part of a feasible instance of \cref{eq:sdpcbnorm3}, then it automatically has Boolean behaviour of degree $2d$, and any $(u,v,A)$ defines a feasible instance for \cref{eq:sdpcbnorm3}. Finally, by Parseval's identity we can rewrite $\sum_{x\in \{-1,1\}^n} \frac{p(x)\phi(x)}{2^n}$ as $\sum_{S\subset [n]}\hat{p}(S)\hat{\phi}(S)$. Putting altogether we get that $p$ is the output of $d$-query quantum algorithm if and only if its degree is at most $2d$ and
	\begin{align*}
			&\ \sum_{S\subset [n],|S|\leq 2d} \hat{p}(S)\langle u,A(i^S_1)\dots A(i^S_{2d})v\rangle\leq 1\\
			\mathrm{s.t.}\
			&(u,v,A)\ \mathrm{has\ Boolean\ behavior\ of\ degree\ }2d,
	\end{align*}  
	which is the same as saying that $\norm{p}_{\fcb,2d}\leq 1$.
\end{proof}

\section{Aaronson and Ambainis conjecture for (Fourier) completely bounded polynomials}
In this section we prove \cref{theo:AAFCBconjectureForMaxDegreePol} and \cref{theo:AAforBCB}. Both are based on the construction used by Varopoulos to disprove a degree 3 von Neumann's inequality \cite{Varopoulos:1974}. First of all, we recall the expressions of variance and influences of a polynomial $p:\{-1,1\}^n\to \mathbb{R}$. The variance is given by 
$$\var[p]=\sum_{|S|\geq 1}\hat{p}^2(S),$$ 
and the influence of the $i$-th variable by $$\mathrm{Inf}_i [p]=\sum_{S\ni i}\hat{p}^2(S).$$ The maximum influence of $p$ is $\text{MaxInf}[p]:=\max_{i\in[n]}\mathrm{Inf}_i[p].$ One may interpret $\var[p]$ as the total change of $p$ and $\mathrm{Inf}_i[p]$ as the change of $p$ that is due to varying the $i$-th variable.  

\subsection{AA conjecture for block-multilinear completely bounded polynomials}\label{subsec:AAforBCB}

Before proving \cref{theo:AAforBCB}, we shall specify what is a block-multilinear completely bounded polynomial.  A \emph{block-multilinear polynomial} of degree $d$ is a polynomial $p:\{-1,1\}^{n\times d}\to \mathbb{R}$ such that if we divide the variables $x\in \{-1,1\}^{n\times d}$ in $d$ blocks of $n$ coordinates each, then the every of the monomials of $p$ has at most one coordinate of each of the blocks. In other words, the block-multilinear polynomials of degree $d$ are those that can be written as 
\begin{equation}
	p(x_1,\dots,x_d)=\hat p(\emptyset)+\sum_{s\in [d]}\sum_{\substack{\ind b\in [d]^s\\b_1<\dots<b_s}}\sum_{\ind{i}\in[n]^s} \hat p(\{(b_1,i_1),\dots,(b_s,i_s)\}) x_{b_1}(i_1)\dots x_{b_s}(i_s),
\end{equation}
for every $(x_1,\dots,x_d)\in (\{-1,1\}^n)^d.$
For this kind of polynomials, there is a very natural way of evaluating them in matrix inputs, 
\begin{equation}
	p(A_1,\dots,A_d)=\hat p(\emptyset)\Id_m+\sum_{s\in [d]}\sum_{\substack{\ind b\in [d]^s\\b_1<\dots<b_s}}\sum_{\ind{i}\in[n]^s}  \hat p(\{(b_1,i_1),\dots,(b_s,i_s)\}) A_{b_1}(i_1)\dots A_{b_s}(i_s),
\end{equation}
for every $A_s\in (M_m)^n$, $s\in [d]$ and $m\in\mathbb N$. The \emph{completely bounded norm} of a block-multilinear polynomial is defined as
\begin{equation}
	\norm{p}_{\cb}:=\sup\{\norm{p(A_1,\dots,A_d)}:\ m\in\mathbb{N},\ A_s\in (B_m)^n,\ s\in [d]\}.
\end{equation}
Concerning these polynomials, we can show the following. 

\AAforBCB*

\begin{remark}
	With our proof of the homogeneous case of \cref{theo:AAforBCB} we can show that for the case of $p:\{-1,1\}^{n\times d}\to \mathbb{R}$ being a homogeneous degree $d$ block-multinear polynomial we have the following non-commutative root influence inequality 
	\begin{equation}\label{eq:rootinfluence}
		\norm{p}_{\cb}\geq \sum_{i\in [n]}\sqrt{\mathrm{Inf}_{s,i}[p]},
	\end{equation}
	for any $s\in [d]$. This improves \cite[Theorem 1.4]{bansal2022influence} in two ways. First, we can allow $s$ to be any number in $[d]$, while they only prove the result of $s\in\{1,d\}$. Second, they prove a weaker statement that depends on $d$, namely,  $$\norm{p}_{\cb}\geq \sum_{i\in [n]}\frac{\sqrt{\mathrm{Inf}_{s,i}[p]}}{\sqrt{e(d+1)}},$$ for $s\in \{1,d\}$.
\end{remark}

\begin{remark}
	Given that $p(x_1,\dots,x_d)=x_1(1)\dots x_d(1)$ is a homogeneous degree $d$ block-multilinear completely bounded polynomial with $\var [p]^2=\mathrm{MaxInf}[p]=1$, we have that the homogeneous case of Theorem \ref{theo:AAforBCB} is optimal.  
\end{remark}

\begin{proof}[ of the homogeneous degree $d$ case of \cref{theo:AAforBCB}]
	Let $p$ be a homogeneous degree $d$ block-multilinear polynomial. Let $s\in [d]$. We label the coordinates by $(r,i)$, where $r\in [d]$ indicates the block, and $i\in [n]$.  Our goal is defining $A\in (B_m)^n$ and $f_\emptyset, e_\emptyset\in S^{m-1}$ such that 
	\begin{equation}\label{eq:desiredcorrelations}
		\langle f_{\emptyset}, A(i_1)\dots A(i_d)e_{\emptyset}\rangle=\frac{\hat{p}(\{(1,i_1),\dots,(d,i_d)\})}{\sqrt{\mathrm{Inf}_{s,i_s}[p]}}.
	\end{equation}
	Once we are there, we can prove the announced root-influence inequality \cref{eq:rootinfluence}. Indeed, 
	\begin{align*}
		\norm{p}_{\cb}&\geq \sum_{i_1,\dots,i_d\in [n]}\hat{p}(\{(1,i_1),\dots,(d,i_d)\})\langle f_{\emptyset}, A(i_1)\dots A(i_d)e_{\emptyset}\rangle\\
		&=\sum_{i_1,\dots,i_d\in [n]}\hat{p}(\{(1,i_1),\dots,(d,i_d)\})\frac{\hat{p}((1,i_1),\dots,(d,i_d))}{\sqrt{\mathrm{Inf}_{s,i_s}[p]}}\\
		&=\sum_{i\in [n]}\sqrt{\mathrm{Inf}_{s,i}[p]}.
	\end{align*}
	Finally, the statement about the maximal influence quickly follows from the root-influence inequality
	\begin{align*}
		\norm{p}_{\cb}\geq \sum_{i\in [n]}\sqrt{\mathrm{Inf}_{s,i}[p]}\geq \sum_{i\in [n]}\frac{\mathrm{Inf}_{s,i}[p]}{\sqrt{\mathrm{MaxInf}[p]}}=\frac{\var [p]}{\sqrt{\mathrm{MaxInf}[p]}},
	\end{align*}
	which after rearranging yields 
	\begin{equation*}
		\mathrm{MaxInf}[p]\geq \left(\frac{\var [p]}{\norm{p}_{\cb}}\right)^2.
	\end{equation*}

	Hence, it suffices to design $(f_\emptyset,e_\emptyset,A)\in S^{m-1}\times S^{m-1}\times (B_m)^n$ satisfying \cref{eq:desiredcorrelations}. Let $\mathcal{S}:=\{\{(r,i_{r}),\dots,(d,i_{d})\}:i_r,\dots,i_{d}\in [n],\ s+1\leq r\leq d\}$ and $\mathcal{S}':=\{\{(1,i_{1}),\dots,(r,i_{r})\}:i_{1},\dots,i_{r}\in [n],\ r\leq s-1\}$. Let $m:=2+|\mathcal{S}|+|\mathcal{S}'|$. Let $\{e_{\emptyset},e_S,f_{\emptyset},f_{S'}:\ S\in\mathcal{S},\ S'\in\mathcal{S}'\}$ be an orthonormal basis of $\mathbb{R}^m$, and define $A(i)\in M_m$ by 
	\begin{align*}
		A(i)e_S&:=e_{S\cup\{(d-|S|,i)\}},\ \mathrm{for}\ 0\leq |S|\leq d-s-1,\ S\in\mathcal{S},\\
		A(i)e_S&:=\sum_{\substack{S'\in\mathcal{S'}\\ |S'|=s-1}}\frac{\hat{p}(S'\cup S\cup \{(s,i)\})}{\sqrt{\mathrm{Inf}_{s,i}(p)}}f_{S'},\ \mathrm{for}\ |S|= d-s,\ S\in\mathcal{S},\\
		A(i)f_{S'}&:=\delta_{(|S'|,i)\in\mathcal{S}'}f_{S'-\{(|S'|,i)\}},\ S'\in\mathcal{S}'.
	\end{align*}

	One can check that $(f_\emptyset,e_\emptyset,A(i))$ satisfies \cref{eq:desiredcorrelations}. This is because the first applications of the $A(i)$'s act like a \emph{creation} operator and the last as \emph{annihilation} operators. The first $d-s$ applications of the $A(i)$'s  create a \emph{superposition} over some vectors with the coefficients that we are interested in (those Fourier coefficients that involve the $d-s+1$ first applications of $A(i)$). The last $s-1$ applications of $A(i)$ act like annihilation operators that lead us to the desired coefficient. 

	Finally, we claim that $A(i)$ are contractions. Given that $\{e_S:\ 0\leq |S|\leq d-s-1,\ S\in\mathcal{S}\}$, $\{e_S:\ |S|= d-s,\ S\in\mathcal{S}\}$ and $\{f_{S'}:\ S'\in\mathcal{S}'\}$ are mapped to orthogonal spaces, we just have to check than when $A(i)$ is a contraction when is restricted to the span of each of these 3 sets. For the first and third sets of vectors that is clear. For the second is true because for any $\lambda\in [n]^{d-s}$
	\begin{align*}
		\norm{A(i)\sum_{\substack{S\in\mathcal{S}\\|S|=d-s}}\lambda_Se_S}&=\norm{\sum_{\substack{S\in\mathcal{S}\\|S|=d-s}}\sum_{\substack{S'\in\mathcal{S}'\\|S'|=s-1}}\frac{\hat{p}(S'\cup S\cup \{(s,i)\})}{\sqrt{\mathrm{Inf}_{s,i}[p]}}\lambda_Sf_{S'}}\\
		&=\sqrt{\frac{\sum_{\substack{S'\in\mathcal{S}'\\|S'|=s-1}}\left(\sum_{\substack{S\in\mathcal{S}\\|S|=d-s}}\hat{p}(S'\cup S\cup \{(s,i)\})\lambda_S\right)^2}{\mathrm{Inf}_{s,i}[p]}}\\
		&\leq\sqrt{\frac{\sum_{\substack{S'\in\mathcal{S}'\\|S'|=s-1}}\left(\sum_{\substack{S\in\mathcal{S}\\|S|=d-s}}\hat{p}(S'\cup S\cup \{(s,i)\})^2\right)\left(\sum_{\substack{S\in\mathcal{S}\\|S|=d-s}}\lambda_S^2\right)}{\mathrm{Inf}_{s,i}[p]}}\\
		&=\sqrt{\frac{\mathrm{Inf}_{s,i}[p]}{\mathrm{Inf}_{s,i}[p]}}\sqrt{\sum_{\substack{S\in\mathcal{S}\\|S|=d-s}}\lambda_S^2}\\
		&=\norm{\sum_{\substack{S\in\mathcal{S}\\|S|=d-s}}\lambda_Se_S},
	\end{align*}
	where in the inequality we have used Cauchy-Schwarz.
\end{proof}

\begin{proof}[ of the general case of \cref{theo:AAforBCB}]
	Let $p:\{-1,1\}^{n\times d}\to \mathbb{R}$ be a block-multilinear degree $d$ polynomial. For every $s\in [d]$, let $p_{=s}$ be its degree $s$ part. Let $D\in [d]$ be such that $\var [p_{=D}]\geq \var [p]/d,$ which exists because $\var [p]=\sum_{s\in [d]}\var [p_{=s}]$. We will now divide the proof in two parts. One is showing that 
	\begin{equation}\label{eq:BCB1}
		\norm{p_{=D}}_{\cb}\leq \norm{p}_{\cb},
	\end{equation} 
	and the other is proving that 
	\begin{equation}\label{eq:BCB2}
		\mathrm{MaxInf}(p_{=D})\geq \left(\frac{\var [p_{=D}]}{\norm{p_{=D}}_{\cb}}\right)^2.
	\end{equation}
	Once we had done that, the result will easily follow:
	\begin{align*}
		\mathrm{MaxInf}(p)&\geq \mathrm{MaxInf}(p_{=D})\geq \left(\frac{\var[ p_{=D}]}{\norm{p_{=D}}_{\cb}}\right)^2\geq \left(\frac{\var[ p]}{d\norm{p}_{\cb}}\right)^2,
	\end{align*}
	where in the second inequality we have used \cref{eq:BCB2}, and in the third we have used \cref{eq:BCB1} and that $\var [p_{=D}]\geq \var [p]/d$. 
	
	First, we prove \cref{eq:BCB1}. Let $B\in B_{d+1}$ be defined by $B:=\sum_{s\in [D]}e_se_{s+1}^T,$ where $\{e_{s}\}_{s\in [D+1]}$ is an orthonormal basis of $\mathbb{R}^{D+1}$. Note that $\langle e_1,B^se_{D+1}\rangle=\delta_{s,D}$ for all $s\in [d]_0$. Hence, 
	\begin{align*}
		\norm{p_{=D}}_{\cb}&=\sup_{\substack{u,v\in S^{m-1},\ A\in (B_m)^n\\ m\in\mathbb N}} \sum_{\substack{\ind b\in [d]^D\\b_1<\dots<b_D}}\sum_{\ind{i}\in[n]^D}&&
		 \hat{p}_{=D}(\{(b_1,i_1),\dots,(b_D,i_D)\}) \langle u,A_{b_1}(i_1)\dots A_{b_D}(i_D) v\rangle\\
		&=\sup_{\substack{u,v\in S^{m-1},\ A\in (B_m)^n\\ m\in\mathbb N}} \sum_{s\in [d]}\sum_{\substack{\ind b\in [d]^s\\b_1<\dots<b_s}}\sum_{\ind{i}\in[n]^s} &&\hat{p}(\{(b_1,i_1),\dots,(b_s,i_s)\})\\
	& &&\cdot \langle u\otimes e_1,[A_{b_1}(i_1)\otimes B]\dots [A_{b_s}(i_s)\otimes B] v\otimes e_{D+1}\rangle\\
	&\leq\norm{p}_{\cb}.&& 
	\end{align*} 
	Second, we prove \cref{eq:BCB2}. Let $\mathcal{S}:=\{\{(b_1,i_1),\dots,(b_{D-1},i_{D-1})\}:\ b_s\in [d],\ b_1<\dots<b_{D-1},\ i_s\in [n],\ s\in [D-1]\}.$ Let $m:=2+|S|$. Let $\{v,f_\emptyset,f_S:S\in\mathcal{S}\}$ be an orthonormal basis of $\mathbb{R}^m$. For $b\in [d],\ i\in [n]$, define $A_b(i)\in M_m$ by 
	\begin{align*}
		A_b(i)v&:=\sum_{\substack{S\in \mathcal{S}\\ |S|=D-1}}\frac{\hat{p}_{=D}(S\cup\{(b,i)\})}{\sqrt{\mathrm{MaxInf}[p_{=D}]}}f_S,\\
		A_b(i)f_{S}&:=\delta_{(b,i)\in S}f_{S-\{(b,i)\}},\ \mathrm{for}\ S\in\mathcal{S}\cup\emptyset.\\
	\end{align*}
	$A_b(i)$ are contractions because they map the vectors of an orthonormal basis to orthogonal vectors without increasing their norms. Note that for $b_1<\dots<b_D$ and $\ind i\in [n]^D$ we have that 
	$$\langle f_\emptyset, A_{b_1}(i_1)\dots A_{b_D}(i_D)v\rangle=\frac{\hat{p}_{=D}(\{(b_1,i_1),\dots,(b_D,i_D)\})}{\sqrt{\mathrm{MaxInf}[p_{=D}]}}.$$ 
	Thus, 
	\begin{align*}
		\norm{p_{=D}}_{\cb}\geq\sum_{\substack{\ind b\in [d]^D\\b_1<\dots<b_D}}\sum_{\ind{i}\in[n]^D}
		\hat{p}_{=D}(\{(b_1,i_1),\dots,(b_D,i_D)\}) \langle f_\emptyset, p(A_1,\dots,A_d) v\rangle=\frac{\var[ p_{=D}]}{\sqrt{\mathrm{MaxInf}[p_{=D}]}},
	\end{align*}
	which after rearranging yields \cref{eq:BCB2}.
\end{proof}

\subsection{AA conjecture for homogeneous Fourier completely bounded polynomials} Finally, we prove a new case of the AA conjecture.

\AAFCBconjectureForMaxDegreePol*
\begin{proof}
	Let $m:=1+{n\choose 0}+\dots+{n\choose d-1}$. Let $\{v,f_{\emptyset},f_{S}: S\subset [n],\ 1\leq |S|\leq d-1\}$ be an orthonormal basis of $\mathbb{R}^m$. Define the matrices $A(i)\in M_m$ as 
	\begin{align*}
		A(i)v&:=\sum_{\substack{S\ni i\\ |S|=d}}\frac{\hat{p}(S)}{\sqrt{\mathrm{MaxInf}[p]}}f_{S-\{i\}},\\
		A(i)f_S&:=\delta_{S\ni i}f_{S-\{i\}},\ \mathrm{for}\ S\subset [n],\ 0\leq |S|\leq  d-1,
	\end{align*}
	for $i\in [n]$ and $A(n+1):=0$. We claim that $(f_{\emptyset},v,A(i))$ has Boolean behavior of degree $d$. $A(n+1)$ is clearly a contraction. For $i\in [n]$, $A(i)$ is a contraction, as it maps vectors of the orthonormal basis to orthogonal vectors without increasing the norm, because $$\norm{A(i)v}^2=\sum_{S\ni i}\frac{\hat{p}(S)^2}{\mathrm{MaxInf}[p]}=\frac{\mathrm{Inf}_i[p]}{\mathrm{MaxInf}[p]}\leq 1.$$ On the other hand, if $S\subset [n]$ satisfies $|S|\leq d-1$, then any $\ind{i}\in [n+1]^d$ with $S_{\ind{i}}=S$ either has a repeated element of $[n]$ or has an appearance of the index $n+1$, which implies that $\langle f_\emptyset, A(i_1)\dots A(i_d) \rangle =0=\hat p (S)$. If $|S|=d$, then any $\ind{i}\in [n+1]^d$ with $S_{\ind{i}}=S$ has $d$ different indices in $[n]$ (corresponding to the elements of $S$), so in that case
	\begin{equation}\label{eq:AAforFCB2}
		\langle f_\emptyset, A(i_1)\dots A(i_d)v \rangle=\frac{\hat{p}(S)}{\sqrt{\mathrm{MaxInf}[p]}}.
	\end{equation}
	Putting everything together we conclude that $(f_\emptyset,v,A(i))$ has Boolean behavior of degree $d$, so 
	\begin{align*}
		\norm{p}_{\fcb,d}&\geq \sum_{S\subset [n]}\hat{p}(S)\langle f_\emptyset, A(i_1)\dots A(i_d)v \rangle=\sum_{ S\subset [n]}\frac{\hat{p}(S)^2}{\sqrt{\mathrm{MaxInf}[p]}}\\
		&=\frac{\var [p]}{\sqrt{\mathrm{MaxInf}[p]}},
	\end{align*}
	where in the first equality we have used \cref{eq:AAforFCB2}.
	After rearranging, the above expression yields $$\mathrm{MaxInf}[p]\geq \left(\frac{\var[ p]
	}{\norm{p}_{\fcb,d}}\right)^2.$$
\end{proof}

\begin{remark}\label{rem:whynot}
	Sadly, we could not extend the proof of \cref{theo:AAFCBconjectureForMaxDegreePol} to the general case. Now, we aim to illustrate what would go wrong with our technique. 
	
	For example, consider a polynomial $p:\{-1,1\}^4\to \mathbb{R}$ with $\deg(p)=2$ and $\norm{p}_{\fcb,4}\leq 1$. Ideally, we would want to define unit vectors $u$ and $v$ and contractions $A(i)$ such that for every $S\subset [4]$ and every $\ind i \in [\ind i^S]$ they satisfied 
	\begin{equation}\label{eq:desiredcorr}
		\langle u, A(i_1)\dots A(i_{4})v \rangle=\frac{\hat{p}(S)}{\sqrt{\mathrm{MaxInf}[p]}}.
	\end{equation}
	If we emulated the strategy of the proof of \cref{theo:AAFCBconjectureForMaxDegreePol}, then $A(1)v$ should be a \emph{normalized} superposition of orthogonal vectors whose amplitudes are all possible $\hat{p}(S_{\ind i})$ that have $i_4=1$. It can be seen that all $\hat{p}(S)$ with $|S|\leq 2$ must be included among these amplitudes. Hence, the \emph{normalizing} factor of $A(1)v$ should be $\sqrt{\var p}$, instead of $\sqrt{\mathrm{MaxInf}(p)}$. Thus, we would reach \begin{equation*}
		\langle u, A(i_1)\dots A(i_{4})v \rangle=\frac{\hat{p}(S)}{\sqrt{\mathrm{Var}[p]}}
	\end{equation*}
	instead of \cref{eq:desiredcorr}, which would lead to $\norm{p}_{\fcb,4}\geq \sqrt{\var p},$ that does not say much about the influences.
\end{remark}
\begin{remark}\label{rem:howtogeneralize}
	However, there might be a different way of, given a polynomial $p$ of degree at most $d$, choosing $(u,v,A)$ with Boolean behavior of degree $d$ such that 
	$$\langle u, A(i_1)\dots A(i_d)v\rangle=\frac{\hat{p}(S_{\ind i})}{\poly (d,\mathrm{MaxInf}[p])},$$
	for any $\ind i\in [n+1]^d$. If that was true, one could copy and paste the proof of \cref{theo:AAFCBconjectureForMaxDegreePol} and conclude \cref{con:AAconjectureFCB}. This reduces \cref{con:AAconjectureFCB} to a question with flavor of tensor networks (at the end of the day we are seeking for an MPS, but caring about the norm of the matrices instead of the dimension) \cite{RevModPhys.93.045003} and almost-quantum correlations (we could also say that we look for $(u,v,A)$ that define a correlation with some symmetry defined by the Boolean behavior, without imposing the symmetries on the $A$) \cite{navascues2015almost}.
\end{remark}
\begin{question}\label{que:generalization}
	Given a polynomial $p$ of degree at most $d$, is there $(u,v,A)\in \mathscr{BB}^d$ such that 	$$\langle u, A(i_1)\dots A(i_d)v\rangle=\frac{\hat{p}(S_{\ind i})}{\poly (d,\mathrm{MaxInf}[p])},$$
	for any $\ind i\in [n+1]^d$?
\end{question}

\textbf{Acknowledgements.} We want to thank Jop Bri\"et, Sander Gribling and Carlos Palazuelos for useful comments, conversations and encouragement. We also want to thank the reviewers for helpful comments.

\bibliographystyle{alphaabbrv}
\bibliography{Bibliography}

\newcommand{\etalchar}[1]{$^{#1}$}
\begin{thebibliography}{CPGSV21}
\expandafter\ifx\csname urlstyle\endcsname\relax
  \providecommand{\doi}[1]{doi:\discretionary{}{}{}#1}\else
  \providecommand{\doi}{doi:\discretionary{}{}{}\begingroup
  \urlstyle{rm}\Url}\fi

\bibitem[AA09]{aaronson2009need}
S.~Aaronson and A.~Ambainis.
\newblock The need for structure in quantum speedups.
\newblock \emph{arXiv preprint arXiv:0911.0996}, 2009.

\bibitem[AA15]{aaronson2015forrelation}
S.~Aaronson and A.~Ambainis.
\newblock Forrelation: A problem that optimally separates quantum from
  classical computing.
\newblock In \emph{Proceedings of the forty-seventh annual ACM symposium on
  Theory of computing}, pages 307--316. 2015.

\bibitem[AAI{\etalchar{+}}16]{Aaronson2015PolynomialsQQ}
S.~Aaronson, A.~Ambainis, J.~Iraids, M.~Kokainis, and J.~Smotrovs.
\newblock Polynomials, quantum query complexity, and {G}rothendieck's
  inequality.
\newblock In \emph{31st Conference on Computational Complexity, {CCC} 2016},
  pages 25:1--25:19. 2016.
\newblock ArXiv:1511.08682.

\bibitem[ABP19]{QQA=CBF}
S.~Arunachalam, J.~Bri{\"{e}}t, and C.~Palazuelos.
\newblock Quantum query algorithms are completely bounded forms.
\newblock \emph{SIAM J.\ Comput}, 48(3):903--925, 2019.
\newblock Preliminary version in ITCS'18.

\bibitem[ACC{\etalchar{+}}22]{austrin2022impossibility}
P.~Austrin, H.~Chung, K.-M. Chung, S.~Fu, Y.-T. Lin, and M.~Mahmoody.
\newblock On the impossibility of key agreements from quantum random oracles.
\newblock In \emph{Advances in Cryptology--CRYPTO 2022: 42nd Annual
  International Cryptology Conference, CRYPTO 2022, Santa Barbara, CA, USA,
  August 15--18, 2022, Proceedings, Part II}, pages 165--194. Springer, 2022.

\bibitem[Amb07]{ambainis2007quantum}
A.~Ambainis.
\newblock Quantum walk algorithm for element distinctness.
\newblock \emph{SIAM Journal on Computing}, 37(1):210--239, 2007.

\bibitem[Amb18]{ambainis2018understanding}
A.~Ambainis.
\newblock Understanding quantum algorithms via query complexity.
\newblock In \emph{Proceedings of the International Congress of Mathematicians:
  Rio de Janeiro 2018}, pages 3265--3285. World Scientific, 2018.

\bibitem[BBC{\etalchar{+}}01]{polynomialmethod}
R.~Beals, H.~Buhrman, R.~Cleve, M.~Mosca, and R.~de~Wolf.
\newblock Quantum lower bounds by polynomials.
\newblock \emph{J. ACM}, 48(4):778–797, 2001.
\newblock ISSN 0004-5411.
\newblock \doi{10.1145/502090.502097}.

\bibitem[BP19]{briet2018failure}
J.~Bri\"{e}t and C.~Palazuelos.
\newblock Failure of the trilinear operator space {G}rothendieck inequality.
\newblock \emph{Discrete Analysis}, 2019.
\newblock Paper No.~8.

\bibitem[BS21]{bansal2021k}
N.~Bansal and M.~Sinha.
\newblock K-forrelation optimally separates quantum and classical query
  complexity.
\newblock In \emph{Proceedings of the 53rd Annual ACM SIGACT Symposium on
  Theory of Computing}, pages 1303--1316. 2021.

\bibitem[BSdW22]{bansal2022influence}
N.~Bansal, M.~Sinha, and R.~de~Wolf.
\newblock Influence in completely bounded block-multilinear forms and classical
  simulation of quantum algorithms.
\newblock \emph{arXiv preprint arXiv:2203.00212}, 2022.

\bibitem[CPGSV21]{RevModPhys.93.045003}
J.~I. Cirac, D.~P\'erez-Garc\'{\i}a, N.~Schuch, and F.~Verstraete.
\newblock Matrix product states and projected entangled pair states: Concepts,
  symmetries, theorems.
\newblock \emph{Rev. Mod. Phys.}, 93:045003, Dec 2021.
\newblock \doi{10.1103/RevModPhys.93.045003}.

\bibitem[DFKO06]{dinur2006fourier}
I.~Dinur, E.~Friedgut, G.~Kindler, and R.~O'Donnell.
\newblock On the fourier tails of bounded functions over the discrete cube.
\newblock In \emph{Proceedings of the thirty-eighth annual ACM symposium on
  Theory of computing}, pages 437--446. 2006.

\bibitem[DMP19]{defant2019fourier}
A.~Defant, M.~Masty{\l}o, and A.~P{\'e}rez.
\newblock On the fourier spectrum of functions on boolean cubes.
\newblock \emph{Mathematische Annalen}, 374(1):653--680, 2019.

\bibitem[FGG07]{farhi2007quantum}
E.~Farhi, J.~Goldstone, and S.~Gutmann.
\newblock A quantum algorithm for the hamiltonian nand tree.
\newblock \emph{arXiv preprint quant-ph/0702144}, 2007.

\bibitem[GL19]{gribling2019semidefinite}
S.~Gribling and M.~Laurent.
\newblock Semidefinite programming formulations for the completely bounded norm
  of a tensor.
\newblock \emph{arXiv preprint arXiv:1901.04921}, 2019.

\bibitem[Gro96]{grover1996fast}
L.~K. Grover.
\newblock A fast quantum mechanical algorithm for database search.
\newblock In \emph{Proceedings of the twenty-eighth annual ACM symposium on
  Theory of computing}, pages 212--219. 1996.

\bibitem[Iva19]{Ivanishvili2019}
P.~Ivanishvili.
\newblock Aaronson-ambainis conjecture.
\newblock
  \url{https://extremal010101.wordpress.com/2019/10/29/aaronson-ambainis-conjecture/},
  2019.

\bibitem[JZ11]{Jain11}
R.~Jain and S.~Zhang.
\newblock The influence lower bound via query elimination.
\newblock \emph{Electronic Colloquium on Computational Complexity - ECCC}, 7,
  02 2011.
\newblock \doi{10.4086/toc.2011.v007a010}.

\bibitem[LR05]{laurent2005semidefinite}
M.~Laurent and F.~Rendl.
\newblock Semidefinite programming and integer programming.
\newblock \emph{Handbooks in Operations Research and Management Science},
  12:393--514, 2005.

\bibitem[LZ22]{lovett2022}
S.~Lovett and J.~Zhang.
\newblock Fractional certificates for bounded functions.
\newblock \emph{Electronic Colloquium on Computational Complexity - ECCC}, 107,
  2022.

\bibitem[Mid05]{midrijanis2005randomized}
G.~Midrijanis.
\newblock On randomized and quantum query complexities.
\newblock \emph{arXiv preprint quant-ph/0501142}, 2005.

\bibitem[Mon12]{montanaro2012some}
A.~Montanaro.
\newblock Some applications of hypercontractive inequalities in quantum
  information theory.
\newblock \emph{Journal of Mathematical Physics}, 53(12):122206, 2012.

\bibitem[NGHA15]{navascues2015almost}
M.~Navascu{\'e}s, Y.~Guryanova, M.~J. Hoban, and A.~Ac{\'\i}n.
\newblock Almost quantum correlations.
\newblock \emph{Nature communications}, 6(1):6288, 2015.

\bibitem[OSSS05]{o2005every}
R.~O'Donnell, M.~Saks, O.~Schramm, and R.~A. Servedio.
\newblock Every decision tree has an influential variable.
\newblock In \emph{46th Annual IEEE Symposium on Foundations of Computer
  Science (FOCS'05)}, pages 31--39. IEEE, 2005.

\bibitem[OZ15]{o2015polynomial}
R.~O'Donnell and Y.~Zhao.
\newblock Polynomial bounds for decoupling, with applications.
\newblock \emph{arXiv preprint arXiv:1512.01603}, 2015.

\bibitem[Pau03]{paulsenoperatoralgebras}
V.~Paulsen.
\newblock \emph{Completely Bounded Maps and Operator Algebras}.
\newblock 02 2003.
\newblock ISBN 9780521816694.
\newblock \doi{10.1017/CBO9780511546631}.

\bibitem[Sho99]{shor1999polynomial}
P.~W. Shor.
\newblock Polynomial-time algorithms for prime factorization and discrete
  logarithms on a quantum computer.
\newblock \emph{SIAM review}, 41(2):303--332, 1999.

\bibitem[Sim97]{Simon}
D.~R. Simon.
\newblock On the power of quantum computation.
\newblock \emph{SIAM Journal on Computing}, 26(5):1474--1483, 1997.
\newblock \doi{10.1137/S0097539796298637}.

\bibitem[SSW21]{sherstov2021optimal}
A.~A. Sherstov, A.~A. Storozhenko, and P.~Wu.
\newblock An optimal separation of randomized and quantum query complexity.
\newblock In \emph{Proceedings of the 53rd Annual ACM SIGACT Symposium on
  Theory of Computing}, pages 1289--1302. 2021.

\bibitem[Tal20]{tal2020towards}
A.~Tal.
\newblock Towards optimal separations between quantum and randomized query
  complexities.
\newblock In \emph{2020 IEEE 61st Annual Symposium on Foundations of Computer
  Science (FOCS)}, pages 228--239. IEEE, 2020.

\bibitem[Var74]{Varopoulos:1974}
N.~T. Varopoulos.
\newblock On an inequality of von {N}eumann and an application of the metric
  theory of tensor products to operators theory.
\newblock \emph{J. Functional Analysis}, 16:83--100, 1974.
\newblock \doi{10.1016/0022-1236(74)90071-8}.

\bibitem[YZ22]{yamakawa2022verifiable}
T.~Yamakawa and M.~Zhandry.
\newblock Verifiable quantum advantage without structure.
\newblock In \emph{2022 IEEE 63rd Annual Symposium on Foundations of Computer
  Science (FOCS)}, pages 69--74. IEEE, 2022.

\end{thebibliography}
\end{document}